\documentclass[12pt,times,letter]{article}
\usepackage{amsthm}
\usepackage{amscd}
\usepackage{amsbsy}
\usepackage{mathtools}
\usepackage{amssymb}
\usepackage{amsmath}
\usepackage{amsfonts}
\usepackage{anysize}
\usepackage{eurosym}
\usepackage{geometry}
\usepackage{ulem}
\usepackage{graphicx}
\usepackage{caption}
\usepackage{color}
\usepackage{setspace}
\usepackage{sectsty}
\usepackage{comment}
\usepackage{footmisc}
\usepackage{caption}
\usepackage{natbib}
\usepackage{pdflscape}
\usepackage{array}
\usepackage{bm} 
\usepackage{multirow}
\usepackage{sgame}
\usepackage{setspace}
\usepackage{tikz}
\usepackage{soul}
\usepackage{ushort} 
\usepackage{array}
\usepackage{appendix}
\usepackage{colortbl}
\usepackage{caption}
\usepackage{subcaption} 
\usepackage{natbib}
\usepackage{caption}
\usepackage{subcaption}
\usetikzlibrary{automata, positioning, arrows}
\usetikzlibrary{decorations.markings, arrows.meta} 
\usetikzlibrary{decorations.markings, arrows.meta}
\usepackage{afterpage}
\usepackage{float}
\usepackage{placeins}

\usetikzlibrary{decorations.pathreplacing, arrows.meta}

\newtheorem{definition}{Definition}
\newtheorem{lemma}{Lemma}
\newtheorem{theorem}{Theorem}

\newtheorem{proposition}{Proposition}
\newtheorem*{remark}{Remark}

\theoremstyle{definition}
\newtheorem{example}{Example}

\newcolumntype{L}[1]{>{\raggedright\let\newline\\arraybackslash\hspace{0pt}}m{#1}}
\newcolumntype{C}[1]{>{\centering\let\newline\\arraybackslash\hspace{0pt}}m{#1}}
\newcolumntype{R}[1]{>{\raggedleft\let\newline\\arraybackslash\hspace{0pt}}m{#1}}

\newcolumntype{P}[1]{>{\centering\arraybackslash}p{#1}} 

\newcommand{\N}{\mathbb{N}}

\newcommand{\R}{\mathbb{R}}

\renewcommand{\P}{\mathbb{P}}
\newcommand{\E}{\mathbb{E}}

\definecolor{Blue}{rgb}{0.00, 0.00, 1.00}

\definecolor{Red}{rgb}{1.00, 0.00, 0.00}

\subsectionfont{\normalsize\mdseries\itshape}
\subsubsectionfont{\normalsize\mdseries\itshape}
\paragraphfont{\normalsize\mdseries\itshape} 

\begin{document}

	\title{A Folk Theorem for Indefinitely Repeated Network Games}
	
	  \author{Andrea Benso\\ 
		Dipartimento di Matematica e Informatica ``Ulisse Dini''\\ 
		Universit\`{a} degli Studi di Firenze \\
		Viale Giovanni Battista Morgagni 67/A, 50134 Florence,  Italy 
	}
	\date{}
	\maketitle
	
	\begin{abstract}
		\noindent We consider a repeated game in which players, considered as nodes of a network, are connected. Each player observes her neighbors' moves only. Thus, monitoring is private and imperfect. Players can communicate with their neighbors at each stage; each player, for any subset of her neighbors, sends the same message to any player of that subset. Thus, communication is local and both public and private. Both communication and monitoring structures are given by the network. The solution concept is perfect Bayesian equilibrium. In this paper we show that a folk theorem holds if and only if the network is 2-connected for any number of players.
	\end{abstract}

	\strut
	
	\noindent\textbf{Keywords:} Folk theorem, Imperfect private monitoring, Networks, Repeated games.
	
	\medskip 
	
	\noindent\textbf{JEL Codes:} C72, C73 
	
	\vspace{7.25cm}
	
	\hrulefill 

	\thanks{*The author acknowledges his membership to the National Group
		for Algebraic and Geometric Structures, and their Applications
		(GNSAGA–INdAM).}
	
	%
	
	\setcounter{page}{0} \thispagestyle{empty} \pagebreak \newpage

\newpage 

\section*{Introduction}
Many papers on folk theorems with imperfect monitoring consider a setting in which players can send messages to all their opponents ({\it global} communication) but are constrained to send the same message to all of them ({\it public} communication). However, in many real situations, communication is {\it local}: each player can communicate only with a subset of players, called {\em neighbors}. In this paper we study repeated games with local monitoring and local communication: players observe only the moves of their neighbors and can communicate only with them. Thus, the monitoring structure is given by a network, where the nodes are the players and edges link neighbors.

The goal is to determine for which networks a folk theorem holds, meaning under which conditions all feasible, strictly individually rational payoffs can be sustained as equilibrium payoffs in the repeated game with high discounting. In particular, we study perfect Bayesian equilibria (PBE) and show that---under the usual non-emptiness assumption on the interior of all feasible, strictly individually rational payoffs---a folk theorem holds if and only if the network is 2-connected.  One of the main challenges is to construct a communication system so reliable as to guarantee the accurate transmission of information---past action deviations---even in presence of false messages; to address this problem we assume that players can send at each stage (costless) messages to any subset of their neighbors and are constrained to send the same message to any player of that subset; in this way, we first show that for any player it is impossible to learn false informations and then that, after an action deviation, all players will learn the deviator's name in finite time. Due to the latter properties of the communication system, we use the approach of \cite{Fudenberg_Maskin_86, Fudenberg_Maskin_91} to construct an equilibrium strategy, divided in four phases: a cooperation phase, a communication phase, a punishment phase and a reward phase; players after an action deviation (which interrupts the cooperation phase), inform other players of the deviator's name (communication phase) and then punish her (punishment phase); finally, players are rewarded for following the strategy (reward phase). The communication protocol we are proposing is similar to the one proposed by Laclau, Renou and Venel in \cite{Laclau_Renou_Venel_2024}: they consider sender-receiver games, where the sender and receiver are two distant nodes in a communication network; therefore, they construct a protocol in which the intermediaries (the players between the sender and the receiver) cannot manipulate the message. The main difference between their protocol and the one proposed in this paper is that in \cite{Laclau_Renou_Venel_2024} the players know when the sender starts to send the message and can therefore organize themselves to build, from the beginning of the game, a scheme (a time partition) in which they analyze the messages received; here, we see the deviator in action as the sender of the message and all other players as the receivers; however, the players who are not neighbors' deviator in action do not know about that deviation and have to find a way to decode the messages that arrive to them. In addition, in our setting players could make false announcements of action deviation of one of their neighbors so that, after punishing her, they are rewarded in the reward phase. Therefore, we need to build a reliable protocol even under this kind of deviation.

This problem connects game theory and computer science. In game theory, players are rational agents who make decisions based on individual preferences; in literature many works have studied how the conflict between equilibrium and social optimality can be solved in repeated games (see \cite{Aumann_Shapley_1994, Fudenberg_Maskin_86, Mailath_Samuelson_06, Rubinstein_77, Sorin_86}). Meanwhile, in computer science, a parallel field studies fault-tolerant communication in distributed systems, where processors communicate over a network; in particular, this field focuses on network topologies that allow reliable communication protocols (see \cite{Dolev_Dwork_Waarts_Yung_93, Franklin_Wright_00, Pease_Shostak_Lamport_80}). In this paper, we study a problem relevant to both fields: designing protocols for players to report their observations and identify non-cooperative behavior. We characterize the network structures where such protocols exist, construct explicit protocols, and apply them to repeated games.  

\medskip

\noindent {\bf Related literature.} 
Analyzing long-term relationships between rational agents is the task of the theory of repeated games with patient players. The existence of efficient equilibrium outcomes has been recognized for a long time, starting with Aumann and Shapley in \cite{Aumann_Shapley_1994}; a key result has been proven by Fudenberg and Maskin in \cite{Fudenberg_Maskin_86, Fudenberg_Maskin_91}, who established the folk theorem for subgame perfect equilibria in discounted repeated games with high discount factors under the assumption of perfect monitoring. Later, many papers on folk theorems with imperfect monitoring focused on imperfect public monitoring (\cite{Fudenberg_Levine_Maskin_94}). In the discounted case, Lehrer (\cite{Lehrer_89, Lehrer_92}) provided a fairly comprehensive study of the equilibrium set for two-player games with imperfect private monitoring. With more than two players, the difficulty relies on the necessity for the players to coordinate their behavior to punish potential deviators. In \cite{Fudenberg_Levine_91}, Fudenberg and Levine established a folk theorem with imperfect private monitoring without explicit communication. They consider private random signals induced by the action profile of all other players. Under the assumption of public and global communication, sufficient conditions for a folk theorem to hold are given in \cite{Compte_98, Kandori_Matsushima_98, Obara_09}; Ben-Porath and Kahneman in \cite{Ben-Porath_Kahneman_96} established a folk theorem for the case in which players observe their neighbors' moves; however, they maintain the assumption of public and global communication. In \cite{Laclau_13}, Laclau extended the result to private communication, maintaining the global communication assumption.  In \cite{Laclau_14}, Laclau studied repeated games with local interaction, where players only interact with their neighbors on a network, and with global communication. With public and local communication, Renault and Tomala in \cite{Renault_Tomala_98} and Tomala in \cite{Tomala_11}, studied repeated games where players observe their neighbors' moves. In \cite{Renault_Tomala_04bis, Renault_Tomala_08} Renault and Tomala studied local communication and in \cite{Laclau_12} Laclau studied repeated games with local interaction and local private communication. Yet, these papers do not impose sequential rationality. 
This work aims to solve this problem, indeed we assume all players may deviate unilaterally. A first result has been proven by Laclau in \cite{Laclau_12_workingpaper} for the case of four-player games, so we generalize the latter to any number of players.

\medskip 

The paper is organized as follows. In Section \ref{sec:the model} is introduced the model. In Section \ref{sec:main_result} is presented the main result. In Section \ref{sec: The equilibrium strategy} is described the strategy. In Section \ref{sec: properties} are provided properties and examples. In Section \ref{sec:proof_main_result} is demonstrated the main result. Finally, are discussed concluding remarks and open questions.

\section{The model} \label{sec:the model}

We consider a repeated game in which players at each period observe only the actions taken by their neighbors in a given graph. Formally, the model is described by the following data:
\begin{itemize}\renewcommand{\labelitemi}{-} 
	\item a finite set of players  $N = \{1, \dots, n\}$, with $ n \geq 2 $.
	\item For each player $ i \in N$, a non-empty finite set $A_i$ of actions, with $|A_i| \geq 2$. Denote $A = \prod_{i \in N} A_i$.
	\item An undirected graph $G = (N, E)$, where the vertices are the players and $E \subseteq \binom{N}{2}$\footnote{For a set $N$ and a positive integer $k=1,\ldots,|N|$, we denote by $\binom{N}{k}$ all the subsets of $N$ with cardinality $k$, i.e., $\{Y \subseteq N\, : \,  |Y|=k\}$. } is the set of edges. 
	\item A payoff function for each player $i \in N$, given by $u_i: A \to \mathbb{R} $.
	\item A (infinite) set of messages $\mathcal{M}_i $ for each player $i$. The description of $\mathcal{M}_i$ is given in the following sections.
\end{itemize}

We use the following notations: for each player $i$, $\mathcal{N}_i = \{ j \neq i \mid \{i, j\} \in E \} $ is the set of neighbors of $i$, $N_{-i}= N \setminus \{i\}$, $E_{-i}=E \setminus \{\{i,j\}\, |\, j \in \mathcal{N}_i \}$,  $A_{\mathcal{N}_i} = \prod_{j \in \mathcal{N}_i} A_j $ and $u = (u_1, \dots, u_n) $ denotes the payoff vector. The repeated game proceeds as follows. At each stage $t \in\N$:  

\begin{itemize}
	\item[-] simultaneously, players choose actions in their action sets and send costless messages to any subset of their neighbors. More precisely, for any player $i$ and subset $X_i \in 2^{\mathcal{N}_i} \setminus \{\varnothing\}$, player $i$ broadcasts the same message to each player in $X_i$; denoting by $\mathcal{M}_{i,X_i}$ the set of messages that $i$ sends to the subset $X_i$, the set $\mathcal{M}_i$ is $\prod_{X_i \in 2^{\mathcal{N}_i} \setminus \{\varnothing\}} \mathcal{M}_{i, X_i}$.  
	\item[-] Let $a^t = (a_i^t)_{i \in N} $ be the action profile at stage $t$. At the end of the stage, each player $i\in N$ observes her neighbors' actions $(a_j^t)_{j \in \mathcal{N}_i} $.
\end{itemize}

Thus, monitoring and communication are determined by the network $G$. Moreover, we assume perfect recall, and that the whole description of the game is common knowledge. For each stage $t$, let $H_i^t $ be the set of private histories of player $i$ up to stage $t$:  
$$
H^t_i = (A_i \times (\mathcal{M}_i)^{2^{\mathcal{N}_i}\setminus \{\varnothing\}} \times (\mathcal{M}_j)^{ 2^{\mathcal{N}_j \setminus \{i\}}\cup \{i\}}_{j \in \mathcal{N}_i} \times \{u_i\} \times A_{\mathcal{N}_i})^t, 
$$
where $\{u_i\}$ is the range of $u_i$ ($H_i^0$ is a singleton). An element of $H^t_i $ is called a $i$-history of length $t$. An \textit{action strategy} for player $i$ is denoted by $\sigma_i = (\sigma_i^t)_{t \geq 1} $, where for each stage $t$, $\sigma_i^t$ is a map from $H_i^{t-1}$ to $\Delta(A_i)$ ($\Delta(A_i) $ denotes the set of probability distributions over $A_i$). A \textit{communication strategy} for player $i$ is denoted by $\varphi_i = (\varphi_i^t)_{t \geq 1}$, where for each stage $t$, $\varphi_i^t$ is a map from $H_i^{t-1}$ to $\Delta((\mathcal{M}_i)^{2^{\mathcal{N}_i}\setminus \{\varnothing\}})$. A \textit{behavior strategy} of player $i$ is a pair $(\sigma_i, \varphi_i) $. Let $\Sigma_i$ be the set of action strategies of player $i$ and $\Phi_i$ be her set of communication strategies. Denote by $\sigma = (\sigma_i)_{i \in N} \in \Sigma:=\prod_{i \in N} \Sigma_i $ the joint action strategy profile of the players and by $\varphi = (\varphi_i)_{i \in N} \in \Phi:= \prod_{i \in N} \Phi_i $ their joint communication strategy. So, given an equilibrium strategy $(\hat{\sigma}, \hat{\varphi})\in \Sigma \times \Phi$, for any player $i$ there can be three types of deviation: in action (from $\hat{\sigma}_i$), in communication (from $\hat{\varphi}_i$) and both together. To ease, if player $i$ deviates from $(\hat{\sigma}_i, \hat{\varphi}_i)$ by choosing $(\tilde{\sigma}_i, \tilde{\varphi}_i)$ ($\neq (\hat{\sigma}_i, \hat{\varphi}_i)$), if $\tilde{\sigma}_i\neq \hat{\sigma}_i$ we will simply say that $i$ is a {\it deviator} when we refer to her as a deviator in action and if $\tilde{\varphi}_i\neq \hat{\varphi}_i$ we will say that $i$ is a {\em liar} when we refer to her as a deviator in communication. Let $H^t$ be the set of histories of length $t$ consisting of sequences of actions, messages and payoffs for $t$ stages and $ H^\infty $ be the set of all possible infinite histories. A strategy profile $(\sigma, \varphi) $ induces a probability distribution $\P_{\sigma, \varphi} $ over the set of plays $H^\infty $, we denote $\E_{\sigma, \varphi}$ the corresponding expectation. 
\medskip  

We consider the discounted infinitely repeated game, where the overall payoff function of each player $i$ is the expected sum of discounted payoffs: 
$$
\gamma_{\delta,i}(\sigma, \varphi) = \E_{\sigma, \phi} \left[ (1 - \delta) \sum_{t=1}^{\infty} \delta^{t-1} u_i(a^t) \right], 
$$
where $\delta \in (0,1) $ is the common discount factor.  
\medskip 

The solution concept is Perfect Bayesian Equilibrium (PBE) of the discounted repeated game. While there is no agreed upon definition of what restrictions apply after histories off the equilibrium path, this plays no role in our construction, and any definition will do. However, we refer to the PBE definition given in \cite[Chapter 4]{Gibbons_92}. More precisely, we construct a restriction of the PBE strategy to a particular class of private histories; namely, the histories along which only unilateral deviations occur. In addition, the construction has the property that, after such histories, the specified play is optimal no matter what beliefs a player holds about their opponents play, provided that the beliefs are such that: for every player $i \in N$, if player $i$ observes a private history compatible with a history along which no deviation has taken place (respectively along which only unilateral deviations have taken
place), then player $i$ believes that no deviation has taken place (respectively only unilateral deviations have taken place). Plainly, this suffices to ensure optimality. Given that play following other kinds of histories (i.e., those involving multilateral deviations) is irrelevant to the outcome, the strategies and beliefs defined in our construction can be 
completed arbitrarily for those off-path histories. Details are in Section \ref{sec:proof_main_result}.
\medskip 	
	
Let $\Gamma_\delta(G, u)$ be the $\delta$-discounted game, and denote by $E^{\text{PBE}}_\delta(G, u)$ the set of PBE payoffs respectively. For each $a\in A$, let $ u(a) = (u_1(a), \ldots, u_n(a)) $ and $ u(A) = \{ u(a) \mid a \in A \} $. The set of feasible payoffs is the convex hull of $u(A)$, denoted by $\text{co}\, u(A)$. The (independent) minmax level of player $i$ is defined by:  
$$
\underline{v}_i = \min_{(x_j)_{j \in N_{-i}} \in \prod_{j \in N_{-i}} \Delta(A_j)}\, \max_{x_i \in \Delta(A_i)} u_i(x_i, (x_j)_{j \in N_{-i}}). 
$$
Henceforth, we normalize the payoffs of the game such that $(\underline{v}_1, \dots, \underline{v}_n) = (0, \dots, 0) $. The set of strictly individually rational payoffs is $IR^*(G, u) = \{ u = (u_1, \dots, u_n) \in \R^n \mid \forall\, i \in N, u_i > 0 \}$. Finally, let $ V^* = \text{co}\, u(A) \cap IR^*(G, u) $ be the set of feasible and strictly individually rational payoffs.
\medskip   

The goal of this paper is to characterize the networks $G$ for which a folk theorem holds, meaning that every feasible and strictly individually rational payoff is a PBE payoff of the repeated game for a sufficiently high discount factor. 

\section{The main result} \label{sec:main_result}
In this setting, players, after observing a deviation from the equilibrium path, start to communicate by sending messages to any subset of their neighbors. 
Our goal is to design a communication protocol that guarantees all players at the end of the protocol to have learnt correctly the name of the deviator even in presence of false messages, so that they can coordinate to punish her, i.e. to force her to her minmax level in the repeated game. We will demonstrate that the information of deviation propagates correctly even in the presence of deviations within the communication mechanism (lies) if and only if the graph is 2-connected, providing thus a necessary and sufficient condition. To do so, denote with $G_{-i}$ the graph obtained by removing player $i \in N$ from the graph $G$. Formally, the resulting graph is defined as: $G_{-i}=(N_{-i},  E_{-i})$.
\begin{definition}\label{2-connected graph}
	A graph $G$ is said to be $2$-connected if for any $i \in N$ the graph $G_{-i}$ is connected.
\end{definition}
In other words, $G$ is 2-connected if removing a vertex results a graph that is still connected, i.e. it is necessary remove at least two vertices to get a disconnected graph. Hassler Whitney (\cite{Whitney1, Whitney2}) provided the following useful characterization for $2$-connected graphs.

\begin{theorem} \label{thm: whitney}
	A graph $G=(N,E)$ with at least three vertices is $2$-connected if and only if for every two vertices $i,j \in
	N$ there exists a cycle on $G$ containing both.
\end{theorem}
We can now state the main result of this paper.

\begin{theorem} \label{main_result}
	Consider the game $\Gamma_\delta(G, u)$ described previously. If the network $G$ is 2-connected and the interior of $V^*$ is nonempty, then for any feasible and strictly individually rational payoffs $v \in V^*$, there exists $\bar{\delta} \in (0,1)$ such that for all $\delta \in (\bar{\delta},1)$, $v$ is a \text{PBE} vector payoff of the $\delta$-discounted game.
\end{theorem} 

Renault and Tomala (\cite{Renault_Tomala_98}) showed that the 2-connectedness of $G$ is a necessary condition. Indeed, if $G$ is not 2-connected, a folk theorem could fail for some payoff function (even one such that the interior of $V^*$ is nonempty). To prove the sufficiency (i.e., Theorem \ref{main_result}) we construct a strategy profile $(\sigma^*, \varphi^*)$ and show that this is a PBE strategy with expected payoff $v \in V^*$ for a discount factor close enough to one. The formal proof is given in Section \ref{sec:proof_main_result}.

\section{The equilibrium strategy} \label{sec: The equilibrium strategy}
The strategy is divided into four phases. First, there is a stream of pure action profiles that yields the desired payoff; this is the equilibrium path. The second phase starts only in case of deviation and is used to inform all players about the deviator's identity and the period in which it occurred; this is described by the communication protocol. Finally, there are the punishment phase and reward phase.

\subsection*{Phase I: equilibrium path}
At each period $t \in \N$ player $i \in N$ chooses $\bar{a}^t_i \in A_i$ such that
$$
(1- \delta) \sum_{t=1}^{\infty} \delta^{t-1} u_i (\bar{a}^t) = v_i; 
$$
this is possible under conditions on discount factor: Sorin (\cite{Sorin_86}) proved that a sufficient condition is $\delta \ge \frac{1}{n}$ and Fudenberg and Maskin (\cite{Fudenberg_Maskin_91}) proved that for every $\varepsilon >0$, there exists $\delta_\varepsilon \in (0,1)$ such that for any $\delta \ge \delta_\varepsilon$ and every $v \in V^*$ such that $v_i \geq \underline{v}_i$ for all $i$, the deterministic sequence of pure actions $\bar{a}^t$ can be constructed so that the continuation payoffs at each stage are within $\varepsilon$ of $v$.
\medskip 

Moreover, at each period $t \in \N$, player $i$ sends the empty message $ (\varnothing, \varnothing, x_{i,t})$ to any subset of her neighbors, where $x_{i,t}$ is uniformly drawn on $[0, 1]$. The meaning of this triplet will be explained in Phase II. 
\subsection*{Phase II: communication phase}

This phase aims at identifying the deviator when a deviation occurs. Formally, the strategy of Phase II can be viewed as a communication protocol: a specification of how players choose their messages, the number of communication rounds, and the output rule. Each player $i \in N$ starts the communication protocol every time detects any kind of deviation from $(\sigma^*, \phi^*)$. For instance, suppose that player $k$ deviates at period $t_0 \in \N$, then player $i$ enters Phase II at ($i$) the end of period $t_0$ if $i \in \{k\} \cup \mathcal{N}_k$, and at $(ii)$ period $t$ if she has received a signal of deviation through the communication protocol at period $t-1$. 
\medskip 

It is worth noting that players may enter Phase II at different stages. Indeed, they start Phase II only when they receive a signal of deviation, which happens at different times for each player. Moreover, players may start a new communication protocol although a previous one has not ended yet. Therefore, there can be several communication protocols running at the same time.
\medskip 

During Phase II, players should stick to the action strategy according to the phase in which the play is. For instance, if players are following the equilibrium path at stage $t$ when they enter Phase II, they should keep playing $\bar{a}$ when performing the protocol, until a possible previous protocol ends up and yields to a new phase of the game.
\medskip 

\begin{center} \text{COMMUNICATION PROTOCOL} \end{center}

Consider the graph $G$ and let $n'$ be the length of the largest cycle on $G$. Suppose that player $k\in N$ deviates at period $t_0 \in \N$. Starting from period $t_0+1$ we consider $n'-3$ blocks, denoted by $B_b$, $b=1,\ldots,n'-3$, of $2n'-3$ periods; crucially, when a player learns about a deviation of $k$ at period $t_0\in \N$, she automatically knows this partition.
\medskip 
\begin{itemize}
\item \textbf{The message space:} the message that player $i\in N$ broadcasts to any subset of her neighbors consists of: $(i)$ a message concerning the deviation that $i$ knows, or the empty message if she doesn't know about a deviation; $(ii)$ a number chosen with uniform probability in $[0,1]$ called \textit{authentication key}; $(iii)$ a sequence of triplets, each of which is composed with the name of a player, a period and a number in $[0,1]$. Formally, the set of messages $\mathcal{M}_{i,X_i}$ that $i$ can broadcast to the subset $X_i \in 2^{\mathcal{N}_i} \setminus \{\varnothing \}$ is 
\begin{align*} 
	\mathcal{M}_{i,X_i} & = \big( ((N \times \N) \cup \{(\varnothing, \varnothing)\}) \times [0,1] \big) \\ 
	&\bigtimes\big( \bigtimes_{j \in \mathcal{N}_i }\big \{\{\varnothing\} \cup (\{j\}\times \N \times [0,1])\big \}\big) \\ 
	&\bigtimes\big( \bigtimes_{j \in N \setminus \mathcal{N}_i }\big \{\{\varnothing\} \cup (\{j\}\times \N \times [0,1])\big \}^L\big), 
\end{align*}
where $L=1+(n'-3)(2n'-3)$; in this way, the first triplet is about the deviation that player $i$ knows; the subsequent $|\mathcal{N}_i|$ triplets are the triplets concerning the neighbors of $i$, these are read as:``{\em $\, z_{j,t}\in [0,1]$ is the random number related to the message of $i$'s neighbor $j \in \mathcal{N}_i$}'', the choice of $z_{j,t}$ will be explained in the following; finally, the last triplets are the triplets sent to $i$ by her neighbors concerning non-neighbors of $i$.  

	
	\item {\bf Transmission of the deviation in action:} 
	\begin{itemize}
		\renewcommand{\labelitemi}{-} 
		\item If $i$ knows about a deviation of player $k \in N$ at period $t_0 \in \N$ at the beginning of the block $B_b=\{t_b,\ldots, t_b + 2n'-4\}$, at each stage $t \in B_b$ broadcasts to her neighbors the triplet $(k,t_0, x_{i,t})$, where  $x_{i,t}$ is uniformly drawn in $[0,1]$;
		
		\item if $i$ doesn't know about a deviation, but has received a message (including the empty message $m_0:=(\varnothing, \varnothing)$)  in the next stage broadcasts the empty message to her neighbors,  and the authentication key $x_{i,t}$.
		
		
	\end{itemize}

	\item {\bf Detection of deviations on communication protocol:} players can deviate from the communication protocol by broadcasting false messages about deviations in action; however, they broadcast it to a whole subset of neighbors, then can be detected. Player $i$ can detect a false message from one of her neighbors $j$ at period $t \in B_b$, in the following cases:\begin{itemize}
		\renewcommand{\labelitemi}{-} 
		\item $i$ knows about a deviation of $k$ at period $t_0$, knows that $j$ knows about it and $j$ broadcast another message; 
		\item $j$ broadcast a message on a deviation that couldn't know about; namely, at period $t \in B_b$ broadcast the triplet $(k,t_0,x_{j,t})$, with $t_0 \in B_{b_0}$ for some $b_0 < b$, and $b-b_0 <d(j,k)$.
		
		
	\end{itemize}
	Then, $i$ broadcasts the triplet $(j,t,z_{j,t})$, where $z_{j,t}= x_{j,t}$ if $i$ detected a lie of $j$ and $z_{j,t}= y$, where $y$ is randomly drawn in $[0,1]$, otherwise. The key observation is that only neighbors of $j$ know the authentication key of $j$, then those who do not know about a deviation can authenticate whether $j$ lied by cross-checking the true authentication key of $j$ with that arrived to them from the other neighbors.
	
	\item {\bf Transmission of past deviations on communication protocol:}
	if $i$ has received the triplet $(j,\tau,z_{j,\tau})$, with $j \notin \{i\} \cup \mathcal{N}_i $, then $i$ broadcasts it to $\mathcal{N}_i$ in the next period. 
	

	\item {\bf Auto-correcting past own deviations:} if $i$ has received the triplet $(j, \tau,z_{j,\tau})$ at stage $t$ but didn’t forward it at stages $t+1$, then she forwards it at stage  $t+2$. 

\item {\bf The decoding rule:}  at the end of each block, the players who don't know about a deviation analyze the messages received. The player analyzes messages as follows: 

\begin{itemize}
	\renewcommand{\labelitemi}{--} 
	\item if $i$ has received, from her neighbor $j$, during the block $n' - 1$ times a message containing the same message about a deviation, let say at stages $\tau_1,\ldots,\tau_{n'-1}$
	\item and if $i$ has not received by stage $\tau_{n'-1}$ at the latest from her other neighbors the message $(j, \tau_1, x_{j,\tau_1})$ where $x_{j,\tau_1}$ matches the authentication key that $i$ received from $j$ at stage $\tau_1$,
\end{itemize}
then, player $i$ learns the message and starts the next block as a player who knows that deviation in action. Otherwise, player $i$ does not learn the message. Moreover, once a player knows the message, she knows it at all the subsequent blocks.
\end{itemize}

\begin{remark}
	{\em It is worth remarking that, by definition, a player who does not know about a deviation does not know the time partition; however, she has received a message about a ``possible'' deviation, meaning that she knows the ``possible'' time partition generated and then when end the blocks.}
\end{remark}

Finally, another possible deviation could be a false announcement about a deviation by a neighbor. More precisely, player $j$ announces, starting then the communication protocol, that her neighbor $i$ has deviated despite $i$ followed the equilibrium path. Therefore, we have to define the strategy in this case as well.  In this situation, all the players $\mathcal{N}_i\cap \mathcal{N}_j$, neighbors of both $i$ and $j$, act as if they have detected a lie by $j$ and the players $\mathcal{N}_j \setminus\{\mathcal{N}_i\cap \mathcal{N}_j\}$ who are neighbors of $j$ and not of $i$ act as agents who do not know about a deviation. Then, all players in $\mathcal{N}_j \setminus\{\mathcal{N}_i\cap \mathcal{N}_j\}$ following the decoding rule will not learn the false announcement of $j$. 


\subsection*{Phase III: punishment phase}

For each player $i \in N$, if $i$ knows that player $k$ deviated at stage $t_0$ and knows that (due to the communication protocol) after $L$ periods every other player knows it too, starting from stage $t_0+L+1$ player $i$ enters a punishment, whose goal is to minmax player $k$. Each player $i \neq k$ plays according to his minmax strategy against $k$, denoted $(P_i(k))$. Denote by $P(k) = (P_i(k))_{i \in N_{-k}}$ the profile of minmax strategies against player $k$. For any strategy $(\sigma_k, \varphi_k)$ of player $k$:
$$
\gamma_{\delta, k} (\sigma_k, P(k), \varphi_k, (\phi_i)_{i \in N_{-k}}) \leq \sum_{t=1}^{\infty} (1-\delta) \delta^{t-1} \underline{v}_k \leq 0,
$$

where $\phi_i$ stands for any communication strategy of each minmaxing player $i \in N_{-k}$.

In particular, a player $j \in N_{-k}$ who chooses an action which is not in the support of her minmax strategy is identified as a deviator. In addition, during the punishment phase, each player $i \in N$, including player $k$, keeps sending messages as in the Phase I. The punishment phase lasts at period $t_0 + L + T(\delta)$ (the construction of $T(\delta)$ is the same as in \cite{Fudenberg_Maskin_86}).

\subsection*{Phase IV: reward phase}

After the punishment phase each player enters a reward phase, whose goal is twofold:

\begin{itemize}
	\item[--] in order to provide each minmaxing player, who is not minmaxed herself, an incentive to play her minmax strategy in Phase III, an additional bonus $\rho > 0$ is added to her average payoff. If the discount factor is large enough, the potential loss during the punishment is compensated by the future bonus;
	
	\item[--] moreover, to induce each minmaxing player to draw her pure actions according to the right distribution of her minmax strategy, we add a phase so that her payoff in the continuation game varies with her realized payoff in a way that makes her indifferent between the pure actions in the support of her minmax strategy. As in \cite{ Fudenberg_Maskin_86, Fudenberg_Maskin_91}, it is convenient to require that any feasible and strictly individually rational continuation payoff must be exactly attained. Otherwise, a minmaxing player might not be exactly indifferent between her pure actions in the support of her minmax strategy.
\end{itemize}

The possibility of providing such rewards relies on the full dimensionality of the payoff set (we assumed that $\text{int}\, V^* \neq \varnothing$). In addition, it is convenient to require that all players know the sequences of pure actions played by each minmaxing player during the punishment phase. For this reason, at the first stage of the reward phase each player $i \in N$ sends all the sequences of pure actions that her neighbors $j \in \mathcal{N}_i$ played (recall that each player observes her neighbors’ actions). The action of player $i$ at this stage is arbitrary. By replacing the first component of the first triplet of a message by the sequences of pure actions played by one’s neighbors in the communication protocol of Phase II, each player knows at the end of stage $t_0+ 2L + T (\delta)$ all the sequences actually played by all players. Finally, the construction of the reward phase as well as the specification of $\mu (\delta)$ is the same as in \cite{Fudenberg_Maskin_91}.

\medskip 

Intuitively, the reward phase ensures that minmaxing players have the right incentives to follow their prescribed strategies in Phase III. When a deviation is detected, the deviator is punished for a finite but sufficiently long period. After this punishment, those who correctly applied the punishment strategy are compensated with a higher future payoff. 

\section{Properties} \label{sec: properties}
In this section we will show two key properties of the protocol. The first states that no player can learn a message incorrectly, that is, if a player learns a message about a deviation of player $k$ at period $t_0$, $k$ at period $t_0$ has indeed deviated (Lemma \ref{imposs_fake_news}). The second states that at least one new player learns about the deviation of $k$ at the end of each block, ensuring then that all players know it after $L$ periods (Lemma \ref{at_least_one_learns}). Importantly, we will show these two properties for each cycle on the graph and then conclude that these hold for all 2-connected graphs by applying Theorem \ref{thm: whitney}. 
\medskip

Indeed, after a deviation, the information of this propagates along each cycle of the graph starting from the deviator, until each player has learnt it. To this purpose, we introduce a new notation. Given a cycle $\mathcal{C}$ on $G$, we fix an orientation\footnote{Note that this choice of notation does not imply that the cycle is oriented in the strict sense: $\mathcal{C}$ remains an undirected cycle; the orientation is introduced uniquely for notational convenience.} (e.g. the counterclockwise and clockwise direction) on the vertices of $\mathcal{C}$ and, for a vertex $i \in \mathcal{C}$, denote by $i^-$ and $i^+$ the vertices immediately before and after $i$ in $\mathcal{C}$ with respect to that orientation.
\medskip 

Specifically, we will use the letter $p$ when we refer to the counterclockwise direction and the letter $q$ when we refer to the clockwise direction. So, for $i \in \mathcal{C}$, to refer to the counterclockwise direction, we will write $i:=p$ and for the vertices immediately to the left and right $p^-$ and $p^+$ respectively. The similar reasoning applies to the clockwise direction. An example is given in Figure \ref{fig:directions_on_cycle}. 
\medskip 
\begin{figure}[h]
	\centering
	\begin{tikzpicture}
		\foreach \i in {1,...,12} {
			\node[circle, draw, fill=black, inner sep=1.5pt] (\i) at ({360/12 * (\i - 1)}:2cm) {};
		}
		
		\foreach \i in {1,...,12} {
			\pgfmathtruncatemacro{\next}{ifthenelse(\i<12, \i+1, 1)}
			\draw (\i) -- (\next);
		}
		
		\node[left] at (5) {$p^-$};
		\node[left] at (6) {$p=i$};
		\node[left] at (7) {$p^+$};
		
		\node[right] at (12) {$j=q$};
		\node[right] at (11) {$q^+$};
		\node[right] at (1) {$q^-$};
	\end{tikzpicture}
	\caption{\label{fig:directions_on_cycle}}
\end{figure}

In this way, we can say that the information propagates along the two (opposite) directions. Therefore, in Lemma \ref{imposs_fake_news} we will show that if an information arrives from a neighbor, if it is false the player would know it due to the authentication key arrived from the other direction of the cycle; meanwhile, in Lemma \ref{at_least_one_learns}, we will show that at each period the information of that deviation propagates along at least one of the two directions.

\begin{lemma} \label{imposs_fake_news}
	Suppose players follow the strategy $(\sigma^*, \varphi^*)$. If the graph $G$ is 2-connected and at most one player deviates from $(\sigma^*, \varphi^*)$ at each stage, the following properties hold: 
	\begin{itemize}
	\item[$1$.] If player $k \in N$ deviates in action at period $t_0 \in \N$, it is not possible for any player $i \in N$ to learn about another deviation in action at period $t_0$; 
	\item[$2$.] If player $k \in N$ broadcasts a false announcement about a deviation in action of some player $i$ at period $t_0$, it is not possible for any player to learn that message.
	\end{itemize}
\end{lemma}

\begin{proof} In Appendix \ref{Appendix_Lemma1}.
\end{proof}

\begin{lemma} \label{at_least_one_learns}
	Suppose players follow the strategy $(\sigma^*, \varphi^*)$ and player $k \in N$ deviates in action at period $t_0 \in \N$. If the graph $G$ is 2-connected and at most one player deviates from the communication protocol at each stage, at least one player learns that deviation every $2n'-3$ periods.
\end{lemma}
\begin{proof}
	In Appendix \ref{Appendix_Lemma2}.
\end{proof}
\medspace

Now, due to the properties described by lemmata \ref{imposs_fake_news} and \ref{at_least_one_learns}, we can state the following result.

\begin{proposition}\label{protocol_properties}
	If some player deviates in action and at most one player deviates from the communication protocol at each stage, then each player will know about that deviation in action after at most $L$ periods.
\end{proposition}
\begin{example}
Suppose that players are connected according to the graph $G$ in Figure \ref{fig:graph_example} and that player $k$ deviates. Figure \ref{fig:cycles in the graph} shows how the information about $k$'s deviation spreads simultaneously along each cycle in the graph, following the counterclockwise and clockwise directions. In the figure, the information starts from $k$ and follows the dashed arrows.

\begin{figure}[htbp] 
	\centering
	\begin{tikzpicture} 
		\coordinate (C1) at (0:1.5);
		\coordinate (C2) at (30:1.5);
		\coordinate (C3) at (60:1.5);
		\coordinate (C4) at (90:1.5);
		\coordinate (C5) at (120:1.5);
		\coordinate (C6) at (150:1.5);
		\coordinate (C7) at (180:1.5);
		\coordinate (C8) at (210:1.5);
		\coordinate (C9) at (240:1.5);
		\coordinate (C10) at (270:1.5);
		\coordinate (C11) at (300:1.5);
		\coordinate (C12) at (330:1.5);
		
		\coordinate (E1) at (0.100-0.125, 0.076-0.073);
		\coordinate (E2) at (-0.057-0.125, 0.382-0.073);
		\coordinate (E3) at (-0.505-0.125, 0.642-0.073);
		\coordinate (E4) at  (C6); 
		\coordinate (E5) at (-1.964-0.125, 0.876-0.073);
		\coordinate (E6) at (-2.754-0.125, 0.815-0.073);
		\coordinate (E7) at (-3.423-0.125, 0.642-0.073);
		\coordinate (E8) at (-3.871-0.125, 0.382-0.073);
		\coordinate (E9) at (-4.028-0.125, 0.076-0.073);
		\coordinate (E10) at (-3.871-0.125, -0.230-0.073);
		\coordinate (E11) at (-3.423-0.125, -0.490-0.073);
		\coordinate (E12) at (-2.754-0.125, -0.663-0.073);
		\coordinate (E13) at (-1.964-0.125, -0.724-0.073);
		\coordinate (E14) at  (C8); 
		\coordinate (E15) at (-0.505-0.125, -0.490-0.073);
		\coordinate (E16) at (-0.057-0.125, -0.230-0.073);	
		\foreach \i in {1,...,11} {
			\pgfmathtruncatemacro{\next}{\i+1}
			\draw (C\i) -- (C\next);
		}
		\draw (C12) -- (C1); 
		
		\foreach \i in {1,...,15} {
			\pgfmathtruncatemacro{\next}{\i+1}
			\draw (E\i) -- (E\next);
		}
		\draw (E16) -- (E1); 
		
		\foreach \i in {1,...,12} {
			\fill (C\i) circle (1.5pt);
		}
		\foreach \i in {1,...,16} {
			\fill (E\i) circle (1.5pt);
		}
		
		\node[above] at (C3) {$k$};
	\end{tikzpicture}
	\caption{\label{fig:graph_example}}
\end{figure}

\begin{figure}[htbp] 
	\centering
	\begin{subfigure}{0.323\textwidth}
		\centering
		\resizebox{\textwidth}{!}{
			\begin{tikzpicture}
				\coordinate (C1) at (0:1.5);
				\coordinate (C2) at (30:1.5);
				\coordinate (C3) at (60:1.5);
				\coordinate (C4) at (90:1.5);
				\coordinate (C5) at (120:1.5);
				\coordinate (C6) at (150:1.5);
				\coordinate (C7) at (180:1.5);
				\coordinate (C8) at (210:1.5);
				\coordinate (C9) at (240:1.5);
				\coordinate (C10) at (270:1.5);
				\coordinate (C11) at (300:1.5);
				\coordinate (C12) at (330:1.5);
				
				\coordinate (E1) at (0.100-0.125, 0.076-0.073);
				\coordinate (E2) at (-0.057-0.125, 0.382-0.073);
				\coordinate (E3) at (-0.505-0.125, 0.642-0.073);
				\coordinate (E4) at (C6);
				\coordinate (E5) at (-1.964-0.125, 0.876-0.073);
				\coordinate (E6) at (-2.754-0.125, 0.815-0.073);
				\coordinate (E7) at (-3.423-0.125, 0.642-0.073);
				\coordinate (E8) at (-3.871-0.125, 0.382-0.073);
				\coordinate (E9) at (-4.028-0.125, 0.076-0.073);
				\coordinate (E10) at (-3.871-0.125, -0.230-0.073);
				\coordinate (E11) at (-3.423-0.125, -0.490-0.073);
				\coordinate (E12) at (-2.754-0.125, -0.663-0.073);
				\coordinate (E13) at (-1.964-0.125, -0.724-0.073);
				\coordinate (E14) at (C8);
				\coordinate (E15) at (-0.505-0.125, -0.490-0.073);
				\coordinate (E16) at (-0.057-0.125, -0.230-0.073);
				
				\foreach \i in {1,...,15} {
					\pgfmathtruncatemacro{\next}{\i+1}
					\draw (E\i) -- (E\next);
				}
				\draw (E16) -- (E1); 

				
				\draw[dashed, ->] (C3) -- (C4);
				\draw[dashed, ->] (C4) -- (C5);
				\draw[dashed, ->] (C5) -- (C6);
				\draw[dashed, ->] (C6) -- (C7);
				\draw[dashed, ->] (C7) -- (C8);
				\draw[dashed, ->] (C8) -- (C9);
				\draw[dashed, ->] (C9) -- (C10);
				
				\draw[dashed, ->] (C3) -- (C2);
				\draw[dashed, ->] (C2) -- (C1);
				\draw[dashed, ->] (C1) -- (C12);
				\draw[dashed, ->] (C12) -- (C11);
				\draw[dashed, ->] (C11) -- (C10);
				
				\foreach \i in {1,...,2} {
					\fill (C\i) circle (0.5pt); 
				}
				\foreach \i in {4,...,12} {
					\fill (C\i) circle (0.5pt); 
				}
				\foreach \i in {3} {
					\fill (C\i) circle (1.5pt); 
				}
				
				\foreach \i in {1,...,3} {
					\fill (E\i) circle (1.5pt);
				}
				\foreach \i in {5,...,13} {
					\fill (E\i) circle (1.5pt);
				}
				\foreach \i in {15,...,16} {
					\fill (E\i) circle (1.5pt);
				}

				\node[above] at (C3) {$k$};
				\end{tikzpicture}
			
		}
	\end{subfigure}
	\begin{subfigure}{0.323\textwidth}
		\centering
		\resizebox{\textwidth}{!}{
			\begin{tikzpicture}
				\coordinate (C1) at (0:1.5);
				\coordinate (C2) at (30:1.5);
				\coordinate (C3) at (60:1.5);
				\coordinate (C4) at (90:1.5);
				\coordinate (C5) at (120:1.5);
				\coordinate (C6) at (150:1.5);
				\coordinate (C7) at (180:1.5);
				\coordinate (C8) at (210:1.5);
				\coordinate (C9) at (240:1.5);
				\coordinate (C10) at (270:1.5);
				\coordinate (C11) at (300:1.5);
				\coordinate (C12) at (330:1.5);
				
				\coordinate (E1) at (0.100-0.125, 0.076-0.073);
				\coordinate (E2) at (-0.057-0.125, 0.382-0.073);
				\coordinate (E3) at (-0.505-0.125, 0.642-0.073);
				\coordinate (E4) at (C6);
				\coordinate (E5) at (-1.964-0.125, 0.876-0.073);
				\coordinate (E6) at (-2.754-0.125, 0.815-0.073);
				\coordinate (E7) at (-3.423-0.125, 0.642-0.073);
				\coordinate (E8) at (-3.871-0.125, 0.382-0.073);
				\coordinate (E9) at (-4.028-0.125, 0.076-0.073);
				\coordinate (E10) at (-3.871-0.125, -0.230-0.073);
				\coordinate (E11) at (-3.423-0.125, -0.490-0.073);
				\coordinate (E12) at (-2.754-0.125, -0.663-0.073);
				\coordinate (E13) at (-1.964-0.125, -0.724-0.073);
				\coordinate (E14) at (C8);
				\coordinate (E15) at (-0.505-0.125, -0.490-0.073);
				\coordinate (E16) at (-0.057-0.125, -0.230-0.073);
				
				\foreach \i in {1,...,3} {
					\pgfmathtruncatemacro{\next}{\i+1}
					\draw (E\i) -- (E\next);
				}
				\foreach \i in {14,...,15} {
					\pgfmathtruncatemacro{\next}{\i+1}
					\draw (E\i) -- (E\next);
				}
				\draw (E16) -- (E1); 
				
				\draw (C6) -- (C7);
				\draw (C7) -- (C8);
				
				\draw[dashed, ->] (C3) -- (C4);
				\draw[dashed, ->] (C4) -- (C5);
				\draw[dashed, ->] (C5) -- (C6);
				
				\draw[dashed, ->] (E4) -- (E5);
				\draw[dashed, ->] (E5) -- (E6);
				\draw[dashed, ->] (E6) -- (E7);
				\draw[dashed, ->] (E7) -- (E8);
				\draw[dashed, ->] (E8) -- (E9);
				\draw[dashed, ->] (E9) -- (E10);
				\draw[dashed, ->] (E10) -- (E11);
				\draw[dashed, ->] (E11) -- (E12);
				\draw[dashed, ->] (E12) -- (E13);
				\draw[dashed, <-] (E13) -- (E14);

				\draw[dashed, ->] (C3) -- (C2);
				\draw[dashed, ->] (C2) -- (C1);
				\draw[dashed, ->] (C1) -- (C12);
				\draw[dashed, ->] (C12) -- (C11);
				\draw[dashed, ->] (C11) -- (C10);
				\draw[dashed, ->] (C10) -- (C9);
				\draw[dashed, ->] (C9) -- (C8);

				\foreach \i in {1,...,3} {
					\fill (C\i) circle (0.5pt); 
				}
				\foreach \i in {4,...,12} {
					\fill (C\i) circle (0.5pt); 
				}
				\foreach \i in {3, 7} {
					\fill (C\i) circle (1.5pt); 
				}
				
				\foreach \i in {1,...,3} {
					\fill (E\i) circle (1.5pt);
				}
				\foreach \i in {5,...,13} {
					\fill (E\i) circle (0.5pt);
				}
				\foreach \i in {15,...,16} {
					\fill (E\i) circle (1.5pt);
				}

				\node[above] at (C3) {$k$};
			\end{tikzpicture}
		}
	\end{subfigure}
	\begin{subfigure}{0.323\textwidth}
		\centering
		\resizebox{\textwidth}{!}{
			\begin{tikzpicture}
				
				\coordinate (C1) at (0:1.5);
				\coordinate (C2) at (30:1.5);
				\coordinate (C3) at (60:1.5);
				\coordinate (C4) at (90:1.5);
				\coordinate (C5) at (120:1.5);
				\coordinate (C6) at (150:1.5);
				\coordinate (C7) at (180:1.5);
				\coordinate (C8) at (210:1.5);
				\coordinate (C9) at (240:1.5);
				\coordinate (C10) at (270:1.5);
				\coordinate (C11) at (300:1.5);
				\coordinate (C12) at (330:1.5);
				
				\coordinate (E1) at (0.100-0.125, 0.076-0.073);
				\coordinate (E2) at (-0.057-0.125, 0.382-0.073);
				\coordinate (E3) at (-0.505-0.125, 0.642-0.073);
				\coordinate (E4) at (C6);
				\coordinate (E5) at (-1.964-0.125, 0.876-0.073);
				\coordinate (E6) at (-2.754-0.125, 0.815-0.073);
				\coordinate (E7) at (-3.423-0.125, 0.642-0.073);
				\coordinate (E8) at (-3.871-0.125, 0.382-0.073);
				\coordinate (E9) at (-4.028-0.125, 0.076-0.073);
				\coordinate (E10) at (-3.871-0.125, -0.230-0.073);
				\coordinate (E11) at (-3.423-0.125, -0.490-0.073);
				\coordinate (E12) at (-2.754-0.125, -0.663-0.073);
				\coordinate (E13) at (-1.964-0.125, -0.724-0.073);
				\coordinate (E14) at (C8);
				\coordinate (E15) at (-0.505-0.125, -0.490-0.073);
				\coordinate (E16) at (-0.057-0.125, -0.230-0.073);
				
				\foreach \i in {4,...,13} {
					\pgfmathtruncatemacro{\next}{\i+1}
					\draw (E\i) -- (E\next);
				}
								
				\draw (C6) -- (C7);
				\draw (C7) -- (C8);
				
				\draw[dashed, ->] (C3) -- (C4);
				\draw[dashed, ->] (C4) -- (C5);
				\draw[dashed, ->] (C5) -- (C6);
				
				\draw[dashed, ->] (E4) -- (E3);
				\draw[dashed, ->] (E3) -- (E2);
				\draw[dashed, ->] (E2) -- (E1);
				\draw[dashed, ->] (E1) -- (E16);
				\draw[dashed, ->] (E15) -- (E16);
				\draw[dashed, ->] (E14) -- (E15);

				\draw[dashed, ->] (C3) -- (C2);
				\draw[dashed, ->] (C2) -- (C1);
				\draw[dashed, ->] (C1) -- (C12);
				\draw[dashed, ->] (C12) -- (C11);
				\draw[dashed, ->] (C11) -- (C10);
				\draw[dashed, ->] (C10) -- (C9);
				\draw[dashed, ->] (C9) -- (C8);

				\foreach \i in {1,...,3} {
					\fill (C\i) circle (0.5pt); 
				}
				\foreach \i in {4,...,12} {
					\fill (C\i) circle (0.5pt); 
				}
				\foreach \i in {3, 7} {
					\fill (C\i) circle (1.5pt); 
				}
				
				\foreach \i in {1,...,3} {
					\fill (E\i) circle (0.5pt);
				}
				\foreach \i in {5,...,13} {
					\fill (E\i) circle (1.5pt);
				}
				\foreach \i in {15,...,16} {
					\fill (E\i) circle (0.5pt);
				}

				\node[above] at (C3) {$k$};
			\end{tikzpicture}
		}
	\end{subfigure}
	\caption{\label{fig:cycles in the graph}}
\end{figure}


\end{example}

\section{Proof of the main result} \label{sec:proof_main_result}
In this section, we give the proof of Theorem \ref{main_result}. In particular, we show that for each belief $\mu$, the assessment $((\sigma^*, \varphi^*), \mu)$ is a PBE. Formally, we denote by $H_i^t(UD |(\sigma^*, \phi^*))$ the set of private histories for player $i$ such that: either no deviation (from $ (\sigma^*, \phi^*)$) has occurred, or only unilateral deviations have taken place. That is to say, for any history in $H^t_i(UD |(\sigma^*, \phi^*))$, no multilateral deviation has occurred. Similarly, denote by $H^t(UD |(\sigma^*, \phi^*))$ the set of total histories along which only unilateral deviations, if any, have taken place. A {\it belief assessment} is a sequence $\mu = (\mu^t_i)_{i \in N, t \geq 1}$ with $\mu^t_i : H^t_i \to \Delta(H^t)$:  
given a private history $h^t_i$ of player $i$, $\mu^t_i (h_i)$ is the probability distribution representing the belief that player $i$ holds on the full history. An {\it assessment} is an element $((\sigma, \phi), \mu)$ where $(\sigma, \phi)$ is a strategy profile and $\mu$ a belief assessment. We consider a restricted set of beliefs $\mathcal{B}=(\mathcal{B}_i)_{i \in N}$, which is strictly included in the set of total histories $H^t$. Namely, for each player $i\in N$, every belief in $\mathcal{B}_i$ only assigns positive probability to histories that differ from equilibrium play, $(\sigma^*, \phi^*)$, in as much as, and to the minimal extent which, their private history dictates that it does. Formally, for every player $i \in N$ and every history $h^t_i \in H^t_i$, we denote by $H^t[h^t_i] \subset H^t$ the set of total histories for which the projection on $H^t_i$ is $h^t_i$. A total history $h^t$ in $H^t[h^t_i]$ is said to be \textit{compatible} with private history $ h^t_i$ of player $i$. Now, for every player $i \in N$ and every history $ h^t_i \in H^t_i$, let $H^t[h^t_i](UD |(\sigma^*, \phi^*)) \subseteq H^t[h^t_i]$ be the set containing all the total histories that are compatible with $h^t_i$ and included in $H^t(UD |(\sigma^*, \phi^*))$. Then, for each player $i \in N$ and every history $h^t_i \in H^t_i$:  
$$H^t[h^t_i](UD |(\sigma^*, \phi^*)) = H^t[h^t_i] \cap H^t(UD |(\sigma^*, \phi^*)).$$ The set of beliefs $\mathcal{B}_i$ is then the following:  
$$\mathcal{B}_i = \{ (\mu^t_i)_{t \geq 1} : \forall t \geq 1, \forall h^t_i \in H^t_i, h^t_i \in H^t_i(UD |(\sigma^*, \phi^*)) \Rightarrow \text{supp}(\mu^t_i(h^t_i)) \subseteq H^t[h^t_i](UD |(\sigma^*, \phi^*)) \},$$ 

where $\text{supp}$ stands for the support of $\mu^t_i(h^t_i)$. In other words, the beliefs of any arbitrary player $i \in N$ are such that: if $i$ observes a history compatible with either no deviation or unilateral deviations, then she assigns probability one to the fact that the total history is in $H^t(UD |(\sigma^*, \phi^*))$ and is compatible with $h^t_i$; otherwise, she assigns probability one to the fact that the total history is not in $H^t(UD |(\sigma^*, \phi^*))$; in the latter case, we prescribe the strategy as the player taking the action that maximizes her payoff in the one-shot game and sends the empty message at each period.

\begin{proof}[Proof of Theorem \ref{main_result}]
Consider the strategy $(\sigma^*, \varphi^*)$ constructed above for all private histories in $H^t_i(UD|(\sigma^*, \varphi^*))$, for each player $i \in N$. Take a belief assessment $\mu \in \mathcal{B}$. We now prove that $((\sigma^*, \varphi^*), \mu)$ is a PBE. Suppose that player $k$ deviates or lies at period $t_0 \in \N$. Due to Lemma \ref{imposs_fake_news}, lies do not change continuation strategies, then we focus only on player $k$'s deviations in action. Due to Proposition \ref{protocol_properties}, after $L$ periods all players know about that deviation, then $k$ is minmaxed starting from period $t_0 + L + 1$. Now, take $v'=(v'_1,\ldots,v'_n)$ in the interior of $V^*$ such that $v_i > v'_i$ for all $i$. Since $v'$ is in the interior of $V^*$, there exists $\rho>0$ so that
$$
(v'_1 + \rho, \ldots, v'_{k-1} + \rho, v'_k, v'_{k+1} + \rho, \ldots, v'_n + \rho)
$$
is in $V^*$. 
Let $\pi^k = (\pi^k_1, \dots, \pi^k_n)$ be a joint strategy that realizes these payoffs. Let $m^k = (m^k_1, \dots, m^k_n)$ be an $n$-tuple of strategies such that the strategies for players other than $k$ together minimize player $k$'s maximum payoff, and such that $u_k(m^k) = 0$. Let $w^k_i = u_i(m^k)$ be player $i$'s per-period payoff when minimaxing player $k$. Namely, $m^k$ is the strategy during Phase III, $\pi^k$ the strategy during Phase IV, meanwhile $u(m^k)$ and $u(\pi^k)$ are the related payoffs. For each $i\in N$, choose an integer $T(\delta)_i$ such that
	\begin{equation} \label{eq8_Fudenberg_Maskin_86}
	\frac{\bar{v}_i}{v'_i} < 1 + T(\delta)_i,
\end{equation}
where $\bar{v}_i$ is player $i$'s greatest one-shot payoff. If player $k$ deviates in Phase I, receives at most $\bar{v}_k$ the period she deviates, zero after $L$ stages for $T(\delta)_k$ stages, and $v'_k$ each period thereafter. Her total payoff, therefore, is not greater than  
\begin{equation} \label{eq9_Fudenberg_Maskin_86}
	\bar{v}_k + \delta \frac{1 - \delta^L}{1 - \delta}v_k +  \frac{\delta^{L+ T(\delta)_k +1}}{1 - \delta} v'_k.
\end{equation}
If she follows the equilibrium path, she gets $v_k/(1-\delta)$, so that the gain deviating is less than 
\begin{equation} \label{eq10_Fudenberg_Maskin_86}
	\bar{v}_k - v_k - \delta^{L+1}\frac{1 - \delta^{T(\delta)_k}}{1 - \delta} v'_k.
\end{equation}
Because $\frac{1 - \delta^{T(\delta)_k}}{1 - \delta}$ converges to $ T(\delta)_k$ as $\delta$ tends to $1$, condition \eqref{eq8_Fudenberg_Maskin_86} ensures that \eqref{eq10_Fudenberg_Maskin_86} is negative for all $\delta$ larger than some $\bar{\delta} < 1$. Suppose player $k$ deviates in Phase III when she is being punished, she obtains at most zero the period in which she deviates and then only lengthens her punishment, postponing the positive payoff $v'_k$. Suppose player $k$ deviates in Phase III when another player $k'$ is being punished. Namely, suppose deviates at $s\text{-th}$ period of punishment phase against player $k'$. We have to distinguish the cases in which $\min (T(\delta)_{k'}-s+1,L)=T(\delta)_{k'}-s+1$ and $\min (T(\delta)_{k'}-s+1,L)=L$. Suppose first $\min (T(\delta)_{k'}-s+1,L)=T(\delta)_{k'}-s+1$. Then she receives 
$$
\bar{v}_k + \delta \frac{1-\delta^{T(\delta)_{k'}-s}}{1 - \delta} w^{k'}_{k} + \delta^{T(\delta)_{k'} - s +1} \frac{1 - \delta^{L - T(\delta)_{k'} + s}}{1 - \delta}(v'_k + \rho) + \frac{\delta^{L + T(\delta)_k + 1}}{1-\delta} v'_k;
$$
meanwhile, if she does not deviate receives
\begin{equation*}
\frac{1-\delta^{T(\delta)_{k'}-s + 1}}{1 - \delta} w^{k'}_{k} + \delta^{T(\delta)_{k'}-s + 1} \frac{1}{1 - \delta} (v'_k + \rho).
\end{equation*}
Thus, the gain to deviating is at most
\begin{equation} \label{eq11_Fudenberg_Maskin_86}
\bar{v}_k - w^{k'}_{k} - \delta^{L+1}\frac{1 - \delta^{T(\delta)_k}}{1 - \delta} v'_k - \frac{\delta^{L+1}}{1 - \delta} \rho.
\end{equation}
As $\delta \to 1$, the third term in \eqref{eq11_Fudenberg_Maskin_86} remains finite beacuse $(1 - \delta^{T(\delta)_k})/(1-\delta)$ converges to $T(\delta)_k$. But, because $\delta^{L+1}$ converges to 1, the last converges to negative infinity. A very similar argument can be used for the case in which $\min (T(\delta)_{k'}-s+1,L)=L$. Thus, there exists $\bar{\delta}_k < 1$ such that for all $\delta > \bar{\delta}_k$, player $k$ will not deviate in Phase III if the discount factor is $\delta$. Finally, the argument for why players do not deviate in Phase IV is basically the same as that for Phase I.  
\medskip 

To conclude, notice that the proof of the optimality of $(\sigma^*, \varphi^*)$ above does not take into account the beliefs. Indeed, since 
$\mu_k \in \mathcal{B}_k$, for any history in $H^t_k (UD | (\sigma^*, \varphi^*))$, each player $k$ believes that there is either no deviation or only unilateral deviations. The partial strategy $(\sigma^*, \varphi^*)$ prevents player $k$ from deviating, no matter what her beliefs are. Indeed, 
under $\mu_k$, player $k$ believes that if she deviates, it will lead either to her punishment, or to no changes in her continuation payoff (in case of a lie).

\end{proof}

\section*{Concluding remarks and open problems} \label{sec: concluding remarks}




In this paper we consider only unilateral deviations, but what would happen if we consider $k$ deviations at each stage? We conjecture that our main ideas could be extended to this situation. The communication protocol should work with some modifications, but the connectivity of the graph must be increased: we would need a $2k$-connected graph, meaning that for every pair of vertices, there exist $k$ disjoint cycles containing both. However, this would lead to a complete graph for $k = n$. Then one could ask whether there is a more efficient protocol that can handle multiple deviations without too restrictive connection constraints on the graph. Another interesting question is to adapt these models to deal with coalitions of deviations. There are many questions that appear: what is a good definition of equilibrium to consider coalitions of deviators? Do deviators share their information in a coalition? A first definition for this kind of equilibria has been given in Rosenthal in \cite{Rosenthal_72} and further developed e.g. by Konishi and Ray in \cite{Konishi_Ray_03} and Vartiainen in \cite{Vartiainen_11}; one of the most appropriate might be the one developed by Laraki in \cite{Laraki_09}; here, a set of {\it permissible coalitions} is a subset $\mathcal{C} \subseteq 2^N$ and a profile is a {\it coalitional-equilibrium} if no coalition in $\mathcal{C}$ has a unilateral deviation that profits all its members. Nash-equilibria (\cite{Nash_50}) correspond to $\mathcal{C} = \{\{i\}, i\in N\}$, Aumann-equilibria (\cite{Aumann_60}) (usually called {\it strong-equilibria}) to $\mathcal{C}= 2^N$. For our model, we could consider $\mathcal{C}$ as the set of all neighbors for each player, i.e. $\mathcal{C}=\{\mathcal{N}_i\, |\, i \in N\}$. However, this problem calls for careful study and is left for future research. 


\bibliographystyle{plain} 
\bibliography{biblio.bib}

\appendix 
\section{Proof of Lemma \ref{imposs_fake_news}} \label{Appendix_Lemma1}
\begin{proof}
	
   Suppose that player $k \in N$ deviates at stage $t_0 \in \N$ and let $\mathcal{C}$ be a cycle of length $n^\mathcal{C}$ on $G$ containing $k$; we will show that for any player in $\mathcal{C}$ it is not possible to learn a message regarding a deviation at stage $t_0$ where the deviator is different from $k$. 
\medskip 

By contradiction, assume that player $i:=p \in \mathcal{C}$ learns about another deviation, i.e. learns about a deviation in action $(k',t_0)$ 
not equal to $(k,t_0)$ at the end of some block $B_b= \{t_b, \ldots, t_b+ 2n'-4\}$, with $b \geq 1$. Without loss of generality, assume that $p$ is the first player to learn it. By definition, $p$ must have $(i)$ received from $p^-$ or $p^+$, suppose $p^-$, the message $(k',t_0)$ in a sequence $t_b= \tau_1 < \tau_2 < \cdots < \tau_{n'-1} \le t_b + 2n'-4$ of $n'-1$ stages and  must have $(ii)$ not received from $p^+$ a triplet with the authentication key of $p^-$ at stage $\tau_1$. Since we are assuming that $p$ is the first player to learn such (false) deviation, it must be that $p^-$ is lying at all stages $\tau_1, \dots, \tau_{n'-1}$. However, under the hypothesis of unilateral deviation, it must be that all the other players in ${\cal C} \setminus \{p^-\}$ are not lying at all $\tau_1, \dots, \tau_{n'-1}$. Therefore, player $p^{--}$ at stage $\tau_2$ at the latest broadcasts to $p^-$ and $p^{---}$ the triplet $(p^-, \tau_1, x_{p^-, \tau_1})$ (that is the exact authentication key of $p^-$ at stage $\tau_1$), player $p^{---}$ broadcasts it at stage $\tau_3$ at the latest, and so on. Since in ${\cal C}$ there are $n^\mathcal{C}-2\le n'-2$  players other than $p$ and $p^-$, player $p^+$ broadcasts the triplet $(p^-, \tau_1, x_{p^-, \tau_1})$ to $p$ and $p^{++}$ at the latest at a stage lower than $\tau_{n'-1}$. Thus, player $p$ matches the authentication key arrived from $p^+$ with the one of $p^-$ at stage $\tau_1$ before the stage $t_b + 2n' -4 $ and does not learn such (false) deviation, contradiction. 
\medskip 

Now, under our assumptions, the graph $G$ is 2-connected, so for any other player $i \in N$ there exists a cycle containing both $k$ and $i$ by applying Theorem \ref{thm: whitney}; thus, doing the same reasoning for every cycle containing both we have the statement. The second part of the statement follows directly from the first.

\end{proof}

\section{Proof of Lemma \ref{at_least_one_learns}} \label{Appendix_Lemma2}

\begin{proof}
	If every player follows the strategy, the information about the deviation $(k,t_0)$ spreads throughout each cycle $\mathcal{C}$ on $G$ along the counterclockwise and clockwise directions starting from $k$. We will show that if at most one player lies at each period, at least one new player learns about the deviation of $k$ in each block, thus demonstrating that if the information does not spread in one direction, it spreads in the other almost surely.
	Consider two players $i, j \in \mathcal{C}$ who know the message of the deviation $(k,t_0)$ and refer to them as $p$ and $q$ respectively. The aim is to prove that either $p^+$ or $q^+$ learns the message at the end of the block $B_b = \{t_b, \ldots, t_b + 2n' - 4\}$. 
	\medskip 
	
		The proof is by contradiction. Assume neither $p^+$ nor $q^+$ learns the message at the end of the block. For any $h \in \{p, p^-, q, q^-\}$ denote by ${\cal L}^h$ the stages of the block where player $h$ is lying. Under the hypothesis of unilateral deviation, the sets ${\cal L}^p$, ${\cal L}^{p^-}$, ${\cal L}^q$ and ${\cal L}^{q^-}$ are pairwise disjoint and so
	\begin{equation} \label{(1) Lemma2_Marie}
		|{\cal L}^p| + |{\cal L}^{p^-}| + |{\cal L}^q| + |{\cal L}^{q^-}| \leq |{\cal L}^p \cup {\cal L}^{p^-} \cup {\cal L}^q \cup {\cal L}^{q^-}| \leq 2n' - 3. 
	\end{equation}
	
	For any subset ${\cal L} \subseteq B_b$, denote by $\bar{{\cal L}}$ its complement $B_b\setminus {\cal L}$. By definition, player $h$ lies at all stages in ${\cal L}^h$ and follows the communication protocol at all stages in $\bar{{\cal L}}^h$.  Let $\bar{{\cal L}}^p := \{\bar{\tau}_1, \ldots, \bar{\tau}_{\alpha}, \ldots, \bar{\tau}_{m_p}\}$, with $\bar{\tau}_{\alpha} < \bar{\tau}_{\alpha+1}$ for all $\alpha$. Notice that $m_p = (2n' - 3) - |{\cal L}^p|$. 
	
	\noindent We have to distinguish two cases: $(i)$ $m_p \ge n'-1$ and $(ii)$ $m_p < n'-1$. Suppose $(i)$. In this case, player $p^+$ has received by $p$ more than $n'-1$ messages about the deviation of $k$; then, by the decoding rule, in order for $p^+$ not to learn the message it must be that $p^+$ has recived from $p^{++}$ the authentication key of $p$ of the first stage of any possible subsequence of $n'-1$ stages of $\bar{{\cal L}}^p$. However, only player $p^-$ and $p^+$ know the correct authentication key of $p$ and for the other player the probability of guessing it is zero; thus, player $p^-$ must have lied at all such periods. Formally, consider all sequences in $\bar{{\cal L}}^p$ of consecutive stages\footnote{In such sequences if $\bar{\tau}_{\ell}$ and $\bar{\tau}_{\ell'}$ are elements of the sequence, so are all $\bar{\tau}_{\ell''}$ satisfying $\bar{\tau}_{\ell} < \bar{\tau}_{\ell''} < \bar{\tau}_{\ell'}$.}  $(\bar{\tau}_{\ell_1}, \ldots, \bar{\tau}_{\ell_{n' - 1}})$ with different starting points.  By construction, there are $(m_p + 1) - (n' - 1) = n' - |{\cal L}^p| - 1$ such sequences. Fix the sequence $(\bar{\tau}_{\ell_1}, \ldots, \bar{\tau}_{\ell_{n' - 1}})$. Player $p^+$ does not learn the deviation $(k,t_0)$ only if player $p^-$ lies by broadcasting the correct authentication key $x_{p, \bar{\tau}_{\ell_1}}$ at some stage $t > \bar{\tau}_{\ell_1}$; therefore, $t \in {\cal L}^{p^-}$. Since there are $n' - |{\cal L}^p| - 1$ such sequences, player $p^-$ must have lied at least $n' - |{\cal L}^p| - 1$ times; that is,
	$$
	|{\cal L}^{p^-}| \geq n' - |{\cal L}^p| - 1.
	$$
	Equivalently, we have
	\begin{equation}\label{(3) Lemma2_Marie}
		|{\cal L}^p| + |{\cal L}^{p^-}| \geq n' - 1.
	\end{equation}

	\bigskip \bigskip

	Suppose now $(ii)$. In this case, we have 
	$$|{\cal L}^p| = (2n' - 3) - m_p > 2n' - 3 - (n' - 1) = n' - 2,$$
	and so $$|{\cal L}^p| \geq n' - 1.$$
	Thus, inequality \eqref{(3) Lemma2_Marie} holds also for case $(ii)$. Applying the same reasoning for players $q$ and $q^-$, we get
	\begin{equation} \label{(4) Lemma2_Marie}
		|{\cal L}^q| + |{\cal L}^{q^-}| \geq n' - 1,
	\end{equation}
	Summing equations \eqref{(3) Lemma2_Marie} and \eqref{(4) Lemma2_Marie}, we obtain that
	$$
	|{\cal L}^p| + |{\cal L}^{p^-}| + |{\cal L}^q| + |{\cal L}^{q^-}| \geq 2n' - 2.
	$$
	Contradiction with Equation \eqref{(1) Lemma2_Marie}. As in Lemma \ref{imposs_fake_news}, under the hypothesis of 2-connectedness of $G$, for every two players $k$ and $i$ there exists a cycle on $G$ that contains both. Then, applying the same reasoning to each cycle of $G$ we obtain the statement.

\end{proof}

\end{document}